\documentclass{IEEEtran}
\usepackage{textcomp}
\usepackage{url}
\usepackage{graphicx}
\usepackage{amssymb,amsfonts, amsmath, amsthm}
\usepackage{mathrsfs}
\usepackage{hyperref}
\newtheorem{theorem}{Theorem}%[section]
\newtheorem{corollary}{Corollary}
\newtheorem{lemma}{Lemma} %for Lemma 1 using package (amsthm)
\usepackage{algorithm,algpseudocode}
\algrenewcommand\algorithmicindent{0.7em}
\usepackage{cite}
\usepackage[labelsep=space, font=small, labelfont=bf]{caption} 
\usepackage{subcaption}
\usepackage{bbold}
\usepackage{bbm}
\usepackage{xcolor}
\newtheorem{assumption}{Assumption}

\DeclareMathOperator{\diag}{diag}
\DeclareMathOperator{\blockdiag}{blockdiag}
\def\BibTeX{{\rm B\kern-.05em{\sc i\kern-.025em b}\kern-.08em
    T\kern-.1667em\lower.7ex\hbox{E}\kern-.125emX}}
\begin{document}
\title{Communication-Efficient Distributed Kalman Filtering using ADMM}
\author{Muhammad Iqbal, Kundan Kumar, and Simo S{\"a}rkk{\"a}, \IEEEmembership{Senior Member, IEEE}
\thanks{M. Iqbal, K. Kumar, and S. S{\"a}rkk{\"a} are with the Department of Electrical Engineering and Automation, Aalto University, Finland (email: muhammad.iqbal@aalto.fi; 
kundan.kumar@aalto.fi; 
simo.sarkka@aalto.fi).}
}
\maketitle
\begin{abstract}
This paper addresses the problem of optimal linear filtering in a network of local estimators, commonly referred to as distributed Kalman filtering (DKF). The DKF problem is formulated within a distributed optimization framework, where coupling constraints require the exchange of local state and covariance updates between neighboring nodes to achieve consensus. To address these constraints, the problem is transformed into an unconstrained optimization form using the augmented Lagrangian method. The distributed alternating direction method of multipliers (ADMM) is then applied to derive update steps that achieve the desired performance while exchanging only the primal variables. Notably, the proposed method enhances communication efficiency by eliminating the need for dual variable exchange.  We show that the design parameters depend on the maximum eigenvalue of the network’s Laplacian matrix, yielding a significantly tighter bound compared to existing results. A rigorous convergence analysis is provided, proving that the state estimates converge to the true state and that the covariance matrices across all local estimators converge to a globally optimal solution. Numerical results are presented to validate the efficacy of the proposed approach.
\end{abstract}

%\IEEEpeerreviewmaketitle

\section{Introduction}\label{sec_intro}

The Kalman filter, introduced by Kalman in 1960 \cite{Kalman1960new}, remains one of the most significant contributions to engineering, with widespread applications in science, medicine \cite{mohite2021optimization}, economics \cite{athans1974importance}, and various engineering domains \cite{sarkka2023bayesian}. Traditionally, a single estimator is employed to infer the state of a dynamical system. However, large-scale systems—such as bridges \cite{morgese2024distributed}, smart grids \cite{kar2014distributed}, forest fire monitoring \cite{olfati2012coupled}, phased-array systems \cite{rashid2023high}, and complex cyber-physical systems \cite{an2019distributed}—demand a network of local estimators for accurate state estimation. In such contexts, distributed algorithms provide scalable, modular, and robust solutions. %\cite{kim2016distributed, knorn2020scalable}. 
The complexity of state estimation is further heightened by the presence of process and measurement noise, particularly when heterogeneous sensor models are involved \cite{khan2008distributing}.

%It is well known that the estimation problem for linear Gaussian model formulated as optimization problem can be optimally solved by the Kalman filter. 
%For linear Gaussian systems, the Kalman filter provides the optimal solution (closed formed solution)
 
Distributed Kalman filtering (DKF) is a solution for estimating the state of a linear dynamical system, along with its associated uncertainty (covariance matrix), observed by a sensor network in the presence of process and measurement noise \cite{olfati2007distributed, carli2008distributed, marelli2018distributed, das2016consensus+, talebi2019distributed, cattivelli2010diffusion, battilotti2024optimal, yan2022distributed}. In DKF, the prediction step is performed using only local information, similar to a single estimator, while the update step incorporates both local information and information exchanged from neighboring nodes to achieve consensus \cite{olfati2007distributed, khan2008distributing, cattivelli2010diffusion}. Numerous studies have developed DKF algorithms for state estimation under uncertainty \cite{olfati2007distributed, cattivelli2010diffusion, ryu2023consensus, luo2024improved, yan2024distributed}.
%\textcolor{blue}{Probably this is correct but is it necessary to state this: For linear Gaussian models, the Kalman filter provides the optimal state estimate and associated covariance. It can be derived either through the Bayesian framework, by recursively updating the posterior distribution, or through the optimization framework, by minimizing the mean square estimation error}.
It is well-known that for linear Gaussian models, one can present optimal filtering problem as a maximum a posterior problem (MAP), using Bayesian framework, to design Kalman filter \cite{sarkka2023bayesian}. Similarly, for linear Gaussian models, distributed state estimation using a sensors network, can be presented as distributed optimization problem using MAP \cite{ryu2019distributed, ryu2023consensus, li2020revisiting} in a distributed Bayesian framework \cite{paritosh2022distributed}.
%It is well-known that for linear Gaussian models, the Kalman filter can be derived either through the Bayesian framework, by recursively updating the posterior distribution, or through the optimization framework, by minimizing the mean square estimation error.  \cite{sarkka2023bayesian}. Similarly, for linear models, the distributed Bayesian framework \cite{paritosh2022distributed, paritosh2023distributed}, which involves solving the maximum a posteriori (MAP) estimation problem, is equivalent to the distributed optimization framework and yields the optimal solution \cite{ryu2023consensus, ryu2019distributed}. 
In \cite{ryu2023consensus, ryu2019distributed}, a dual-ascent method is employed to estimate the state of a dynamical system  and its covariance matrix in a fully distributed manner. An ADMM-based approach is used to derive the update rule for DKF \cite{wang2018distributed, wang2019distributed}. In \cite{wang2018distributed}, the number of dual variables equals the number of edges, leading to increased computational complexity. Additionally, the consensus process slows down due to the selection of a smaller step size in updating the information rate matrix, where the step size is inversely proportional to the degree of a node. Furthermore, \cite{wang2019distributed} assumes a complete graph and requires the exchange of measurement matrices and measurement noise covariances, which is a strong assumption. Additionally, neither \cite{wang2018distributed} nor \cite{wang2019distributed} provide a convergence analysis. In \cite{ryu2019distributed,ryu2023consensus}, both primal and dual variables are exchanged to achieve consensus. The design parameter in \cite{ryu2019distributed} is upper bounded by the inverse of the square of the maximum eigenvalue of the Laplacian matrix multiplied by the norm of the information rate matrix for each edge, leading to a small design parameter that slows the consensus process. Although \cite{ryu2019distributed} improves this upper bound, it is still related to the square of the maximum eigenvalue of the Laplacian matrix.

% Major shortcomings in the current literature are the following:
% \begin{itemize}
%     \item The lack of fundamental formulation that embraces the communication constraints and derive the correction step.
%    % \item The requirement of multiple  sub-iteration in the correction step \cite{wang2019distributed, ryu2019distributed, wang2020factor, qin2020randomized, marelli2021distributed, qian2022consensus, greiff2022distributed, yan2022distributed, hu2022efficient, qian2023observation, ryu2023consensus, qian2024harmonic}.
%     \item In the distributed optimization framework, the exchange of both primal-and dual variables \cite{ryu2019distributed, ryu2023consensus}.  
%     \item The lack of fundamental formulation that embraces the communication constraints and derive the correction step.
% \end{itemize}

The contributions of this paper are as follows:
\begin{itemize}
    \item We introduce a variant of distributed ADMM that does not require the exchange of dual variables in the update step, thus reducing the communication burden.
        \item We derive tight upper bounds on the design parameters used in the update step for the posterior state estimate, enabling local estimators to significantly reduce the consensus error within a few sub-iterations.
   \item In the proposed method, the update of information rate matrix does not require sub-iterations.
    \item  We show that the local estimators at each node provide unbiased estimates as time approaches infinity.
\end{itemize}
%In addition, the distributed ADMM algorithm presented in this paper is different from the one presented in \cite{wang2018distributed, wang2019distributed} because the augmented Lagrangian is designed in this paper, in a unique way that produces an alternative of writing update step, which reduces the number of dual variables to the number of nodes.
In addition, the distributed ADMM algorithm in this paper differs from \cite{wang2018distributed, wang2019distributed} because the augmented Lagrangian is designed in a unique way that produces an update step, which reduces the number of dual variables to the number of nodes.
In contrast, \cite{wang2018distributed} introduces dual variables equal to the number of edges, which increases computational complexity. Furthermore, the consensus process slows down due to the algorithm's structure, particularly in dense networks. In addition, we provide convergence analysis which is not given in \cite{wang2018distributed, wang2019distributed}.
 
 \section{Problem Formulation}
 Consider a network of $N \geq 2$ sensor nodes measuring the output of a discrete-time dynamical system:
\begin{equation}
\begin{split}
        \label{equ:system}
    x_{t+1} =& Fx_t + w_t, \\
    y_t =& Hx_t+ v_t,
\end{split}
\end{equation}
where $x_t \in \mathbb{R}^n$ is the state vector at time $t \in \{0, 1, \ldots, T \} $, and $y_t = [y_{1,t}^{\top}, \cdots, y_{N,t}^{\top}]^{\top} \in \mathbb{R}^m$ is the aggregated measurement vector of all sensors. Each sensor $i \in \{1,2, \ldots, N\}$ provides measurements $y_{i,t} \in \mathbb{R}^{m_i}$, with $F$ as the state-transition matrix and $H= [H^{\top}_1, \cdots, H^{\top}_N]^{\top}$ as the measurement matrix. The process noise $w_t$ and measurement noise $v_t = [v_{1,t}^\top, \cdots, v_{N,t}^\top]^\top$ are zero-mean, and white Gaussian, satisfying the following properties: 
\begin{equation*}
\begin{split}
        \mathbb{E}\{w_tw^{\top}_l\} &= Q \delta_{tl}, \; \; \mathbb{E}\{v_{t}v^{\top}_{l}\} = \bar{R} \delta_{tl}, \\
        \mathbb{E}\{w_tv^{\top}_{i,l}\} &= 0, \; \; i = 1, \ldots, N,
\end{split}
\end{equation*}
where $\mathbb{E}\{\cdot\}$ is the expectation operator, $Q$ and $\bar{R} = \mathrm{diag}\{R_1, \cdots, R_N\}$ are positive definite matrices, and $\delta_{tl}$ is the Kronecker delta. The initial state $x_o \sim \mathcal N(\mathbb{E}\{x_o\}, P_o)$ is uncorrelated with $w_t$ and $v_t$. 

The sensor network is represented as an undirected graph $ {\mathcal G=\left(\mathcal V,\mathcal E \right) }$, where $ \mathcal V=\left\{1, 2,\cdots, N\right\}$ is the set of nodes, and $ \mathcal E \subseteq  \mathcal V \times  \mathcal V$ defines the edges. The graph’s adjacency matrix $A$ satisfies $a_{ij}=1$ if node $i$ is receiving information from node $j$, and $a_{ij}=0$ otherwise.  For undirected graphs, $a_{ij}=a_{ji}$. The set of neighbors of node $i$ is denoted as $\mathcal N_i = \left\lbrace j \; | \; a_{ij} = 1, \; j \in \mathcal V \right\rbrace $.

The objective is to design a distributed algorithm that enables each node to estimate the state $x_t$ of a dynamical system using its local measurements and information from its neighbors. Specifically, we propose a distributed ADMM-based approach for the correction step in DKF. The algorithm ensures that the local estimators at each node satisfy:
\begin{equation*}
    \lim_{t \to \infty} \mathbb{E}\{x_t - \xi_{i,t}\}=0,  \; \;   \lim_{t \to \infty} \lVert P^* - P_{i, t\mid t-1} \rVert_F = 0,
\end{equation*}
where $\xi_{i,t}$ is the posterior estimate of the state at node $i$, $\lVert \cdot \rVert_F$ represents the Frobenius norm, $P_{i, t\mid t-1}$ is the local prior covariance matrix at node $i$, and $P^*$ is the unique positive definite solution to the discrete-time algebraic Riccati equation:
\begin{equation*}
    P^*= FP^*F^{\top}-FP^*H^{\top}(HP^*H^{\top} + \bar{R})^{-1}HP^*F^{\top}+Q,
\end{equation*} 
in the case of centralized Kalman filter \cite{ryu2023consensus}.
\section{Distributed Kalman Filtering}
Observing the state of a complex dynamical system using a single sensor is often impractical. The centralized filtering methods can be used to estimate the state of a complex dynamical system but increases the communication burden and fragility. Instead, multiple sensors can be deployed, without having an anchor node, to observe different parts of the state vector of the physical process. To address this, we relax the assumption that the entire system is observable from a single sensor. In this framework, multiple local estimators collaborate, sharing local information to estimate the entire state vector, while each observes only a partial state. To ensure the sensor network can collectively estimate the state of the system, we impose the following assumptions:
%%%%
% \begin{figure*}[h]
% \centering
%     \includegraphics[width=1.01\linewidth]{Block_diagramn.pdf}
%     \caption{A network of local estimators sharing their local information to agree on a common state estimate and covariance estimate. Each node has a prior Gaussian distribution and a likelihood PDF}
% \label{fig:ring_graph}
% \end{figure*}
%%%%%%%%%%%%%%%
\begin{assumption}
    The pair $(F,H)$ is observable.
\end{assumption}
\begin{assumption}
    The network $\mathcal{G}$ is static and connected.
\end{assumption}
Assumption 1 indicates that $(F, H_i)$ is not necessarily observable for any individual sensor.
These assumptions are standard and can be found in \cite{ryu2023consensus} and the references therein. 

In a distributed estimation setting, the $i$th estimator has access to its own measurement \( y_{i,t} \) at time \( t \), the measurement matrix \( H_i \), and the measurement covariance matrix \( R_i \). The matrices \( F \) and \( Q \) are assumed to be known to all estimators.

In DKF, the prediction step is identical to that of centralized Kalman filtering. The local prediction step for each estimator is given by:
\begin{equation}
\begin{split}
    \hat{x}_{i, t|t-1} &= F \hat{x}_{i, t-1 \mid t-1}, \\
    P_{i,t|t-1} &= F P_{i,t-1 \mid t-1}F^{\top} + Q,
\end{split}
\end{equation}
where \( \hat{x}_{i, t|t-1} \) denotes the predicted mean at time \( t \) based on the posterior mean \( \hat{x}_{i, t-1 \mid t-1} \) at time \( t-1 \). Similarly, \( P_{i,t|t-1} \) is the predicted error covariance matrix of agent \( i \), computed using the posterior error covariance \( P_{i,t-1 \mid t-1} \) at time \( t-1 \).

% In DKF, the prediction step is identical to that of centralized Kalman filtering or Kalman filtering with a single estimator. The local prediction step for each estimator is given by:
% \begin{equation}
% \begin{split}
%     \hat{x}_{i, t|t-1} &= F \hat{x}_{i, t-1 \mid t-1}, \\
%     P_{i,t|t-1} &= F P_{i,t-1 \mid t-1}F^{\top} + Q,
% \end{split}
% \end{equation}
% where \( \hat{x}_{i, t|t-1} \) denotes the predicted state at time \( t \) based on the posterior mean \( \hat{x}_{i, t-1 \mid t-1} \) at time \( t-1 \). Similarly, \( P_{i,t|t-1} \) is the predicted error covariance matrix of Agent \( i \), computed using the posterior error covariance \( P_{i,t-1 \mid t-1} \) at time \( t-1 \).

In the correction step, each estimator updates its state estimate using local information and information exchanged with its neighbors. For a linear Gaussian system, this step corresponds to solving the maximum a posteriori (MAP) estimation problem in a distributed manner \cite{ryu2023consensus}. To develop the solution, we first formulate the MAP problem for the $i$th node by ignoring coupling constraints. This serves as a basis for understanding the local estimation problem, which is later extended to incorporate the coupling constraints for the complete distributed formulation.
The MAP problem for the $i$th node, without coupling constraints, is formulated as:
\begin{equation}
\label{cost_fn_initial}
\hat{x}_{i,t \mid t} = \arg \max_{x_t} p(y_{i, t} \mid x_t) p(x_t \mid y_{i, 1:t-1}),
\end{equation}
where \( p(y_{i, t} \mid x_t) = \mathcal{N}(y_{i, t} \mid H_i x_t, R_i) \) is the likelihood function of the $i$th estimator, and \( p(x_t \mid y_{i, 1:t-1}) = \mathcal{N}(x_t \mid \hat{x}_{i, t \mid t-1}, P_{i, t \mid t-1}) \) is its prior probability density function. Using the monotonicity of the logarithmic function, \eqref{cost_fn_initial} is equivalent to:
\begin{equation}
\label{posterior_pdf_MAP1}
\hat{x}_{i,t \mid t} = \arg \max_{x_t} \ln \left( p(y_{i, t} \mid x_t) p(x_t \mid y_{i, 1:t-1}) \right).
\end{equation}
The terms in \eqref{posterior_pdf_MAP1} are expanded as follows:
\begin{align*}
&\ln p(y_{i, t} \mid x_t) = -\frac{1}{2} (y_{i, t} - H_i x_t)^\top R_i^{-1} (y_{i, t} - H_i x_t) \\
&\quad - \frac{1}{2} \ln((2\pi)^{m_i} \det(R_i)), \\
&\ln p(x_t \mid y_{i, 1:t-1}) = -\frac{1}{2} (\hat{x}_{i, t \mid t-1} - x_t)^\top P_{i, t \mid t-1}^{-1}  (\hat{x}_{i, t \mid t-1} - x_t) \\& \quad 
- \frac{1}{2} \ln((2\pi)^n \det(P_{i, t \mid t-1})).
   \end{align*}
Substituting $\ln p(y_{i, t} \mid x_t)$ and $\ln p(x_t \mid y_{i, 1:t-1})$ in \eqref{posterior_pdf_MAP1}, the optimization problem in \eqref{posterior_pdf_MAP1} takes the following form:
\begin{equation}
\label{optimal_estimate_single}
\hat{x}_{i, t \mid t} = \arg \min_{\xi_{i,t}} f_{i,t}(\xi_{i,t}),
\end{equation}
where \( \xi_{i,t} \in \mathbb{R}^n \) and
\[
f_{i,t}(\xi_{i,t}) = \frac{1}{2} (\hat{x}_{i, t \mid t-1} - \xi_{i,t})^\top P_{i, t \mid t-1}^{-1} (\hat{x}_{i, t \mid t-1} - \xi_{i,t}) 
\]
\[
+ \frac{1}{2} (y_{i,t} - H_i \xi_{i,t})^\top R_i^{-1} (y_{i,t} - H_i \xi_{i,t}).
\]
Sensor networks are typically heterogeneous, making \eqref{optimal_estimate_single} suboptimal for the entire network since the decision variable is local. To achieve network-wide optimality, we reformulate \eqref{optimal_estimate_single} as:
\begin{equation}
\label{problem00}
\begin{aligned}
& \underset{\xi_t}{\text{minimize}} & & \sum_{i=1}^N f_{i,t}(\xi_t),
\end{aligned}
\end{equation}
where \( \xi_t = [\xi_{1,t}^\top, \xi_{2,t}^\top, \ldots, \xi_{N,t}^\top]^\top \). However, since Estimator \( i \) does not have access to the entire \( \xi_t \), we adopt a distributed optimization approach based on \cite[Lemma~3.1]{gharesifard2013distributed} and \cite{ryu2023consensus}:
\begin{equation}
\tag{P1}
\label{problem1}
\begin{aligned}
& \underset{\xi_{1,t}, \ldots, \xi_{N,t}}{\text{minimize}} & & \sum_{i=1}^N f_{i,t}(\xi_{i,t}) \\
& \text{subject to} & & \mathbb{L} \xi_t = 0,
\end{aligned}
\end{equation}
where \( \mathbb{L} = \mathcal{L} \otimes I_n \), \( \mathcal{L} = D - A \), with $\mathcal L$ denoting the graph Laplacian, and \( D \in \mathbb{R}^{N \times N} \) is the degree matrix of \( \mathcal{G} \). The constraint \( \mathbb{L} \xi_t = 0 \) ensures \( \xi_{1,t} = \xi_{2,t} = \cdots = \xi_{N,t} \), as the kernel of \( \mathcal{L} \) is spanned by \( 1_N \).
To express \( f_{i,t}(\xi_{i,t}) \) in compact quadratic form, define \( \textbf{z}_{i,t} = [y_{i,t}; \hat{x}_{i, t|t-1}] \), \( \textbf{H}_i = [H_i; I_n] \), and \( \textbf{S}_{i,t} = \text{diag}(R_i, NP_{i, t|t-1}) \). Then:
\[
f_{i,t}(\xi_t) = \frac{1}{2} (\textbf{z}_{i,t} - \textbf{H}_i \xi_t)^\top \textbf{S}_{i,t}^{-1} (\textbf{z}_{i,t} - \textbf{H}_i \xi_t).
\]
Finally, define \( \textbf{z}_t = [\textbf{z}_{1,t}^\top, \cdots, \textbf{z}_{N,t}^\top]^\top \), \( \bar{\textbf{H}} = \text{diag}(\textbf{H}_1, \cdots, \textbf{H}_N) \), and \( \bar{\textbf{S}}_t = \text{diag}(\textbf{S}_{1,t}, \cdots, \textbf{S}_{N,t}) \). The matrix \( \mathcal{H}_t = \mathbb{1}_N^\top \bar{\textbf{H}}^\top \bar{\textbf{S}}_t^{-1} \bar{\textbf{H}} \mathbb{1}_N \), with \( \mathbb{1}_N = 1_N \otimes I_n \), is symmetric positive definite.

Next, we present the distributed correction step for the covariance update using a distributed optimization framework. Let \( \Omega_{t\mid t} = P_{t\mid t}^{-1} \) and \( P_{t\mid t}^{-1} \hat{x}_{t\mid t} \) represent the information matrix and the information vector, respectively. We also define \( \Omega_{t\mid t-1} = P_{t\mid t-1}^{-1} \). For a single estimator, the information matrix prediction and correction steps are given as:
\begin{equation}
\label{cov_update_single}
\begin{aligned}
    \Omega_{t\mid t-1} &= (F\Omega_{t-1\mid t-1}^{-1}F^\top + Q)^{-1}, \\
    \Omega_{t\mid t} &= \Omega_{t\mid t-1} + H^\top \bar{R}^{-1} H.
\end{aligned}
\end{equation}
The convergence of \( \Omega_{t\mid t-1} \) to \( P^{*-1} \) is established in \cite[Lemma~9.5.1 and Prob.~9.17]{kailath2000linear}. In a distributed setting, the convergence of \( \Omega_{t|t-1} \) to \( P^{*-1} \) remains valid if the global information rate matrix \( H^\top \bar{R}^{-1} H \) is available to each estimator. To achieve this, we solve the following consensus optimization problem:
\begin{equation}
\tag{P2}
\label{problem2}
\begin{aligned}
    & \underset{\theta_1, \ldots, \theta_N}{\text{minimize}}
    & & \frac{1}{2} \sum_{i=1}^N \| N\omega_i^\delta - \theta_i \|^2 \\
    & \text{subject to}
    & & (\mathcal{L} \otimes I_{n_{\text{cov}}}) \theta = 0,
\end{aligned}
\end{equation}
where \( \|\cdot\| \) denotes the Euclidean norm, \( \omega_i^\delta = \mathrm{vech}(H_i^\top R_i^{-1} H_i) \in \mathbb{R}^{n_{\text{cov}}} \), \( \mathrm{vech}(\cdot) \) is the half-vectorization of the symmetric matrix $H_i^{\top}R^{-1}_i H_i$, $n_{cov}=\frac{n(n+1)}{2}$, and $\theta_i \in \mathbb{R}^{n_{cov}}$ is the decision variable of $i$th estimator.
% To solve \eqref{problem2}, we propose the following Lagrangian:
% \begin{equation}
%     \label{lag2}
%     L_{cov}(\theta, \nu) = \frac{1}{2}( N\omega^{\delta}- \theta )^{\top}( N\omega^{\delta}- \theta )+ \nu^{\top}\tilde{\mathbb{L}}\theta+ \frac{\alpha}{2}\lVert \sqrt{\mathbb{L}} \theta_t\rVert
% \end{equation}

%=========================================================%
%\section{Distributed Bayesian Framework}
%=========================================================%
\section{Distributed Kalman Filtering using ADMM}
In this section, we derive the distributed Kalman filter algorithm by solving an optimization problem using distributed ADMM. To achieve this, we solve \eqref{problem1} by considering the following augmented Lagrangian:
\begin{equation}
\label{augumented_lag}
    L_{est,t}(\xi_t, \lambda_t) = \sum_{i=1}^N f_{i,t}(\xi_{i,t}) + \lambda_t^\top \sqrt{\mathbb{L}} \xi_t + \frac{\mu}{2} \lVert \sqrt{\mathbb{L}} \xi_t \rVert^2.
\end{equation}
Taking the gradient of \eqref{augumented_lag} with respect to \( \lambda_t \) and \( \xi_t \), we obtain:
\begin{equation}
\label{quad_aug_lag}
\begin{aligned}
    \nabla_{\lambda_t} L_{est,t} &= \sqrt{\mathbb{L}} \xi_t, \\
    \nabla_{\xi_t} L_{est,t} &= \bar{\textbf{H}}^\top \bar{\textbf{S}}_t^{-1} \bar{\textbf{H}} \xi_t - \bar{\textbf{H}}^\top \bar{\textbf{S}}_t^{-1} \textbf{z}_t + \sqrt{\mathbb{L}} \lambda_t + \mu \mathbb{L} \xi_t.
\end{aligned}
\end{equation}
Using \eqref{quad_aug_lag}, the update step for $\lambda_{t,l}$ and $\xi_{t,l}$ can be written as:
\begin{equation}
\label{ADMM_aux_Kalman0}
\begin{aligned}
    \lambda_{t,l+1} &= \lambda_{t,l} + \alpha K_t^{-1} \sqrt{\mathbb{L}} \xi_{t,l}, \\
    \xi_{t,l+1} &= K_t \big( \bar{\textbf{H}}^\top \bar{\textbf{S}}_t^{-1} \textbf{z}_t - \sqrt{\mathbb{L}} \lambda_{t,l+1} - \mu K_t^{-1} \mathbb{L} \xi_{t,l} \big),
\end{aligned}
\end{equation}
where \( K_t = (\bar{\textbf{H}}^\top \bar{\textbf{S}}_t^{-1} \bar{\textbf{H}})^{-1} \).
Due to the structure of \( \sqrt{\mathbb{L}} \), the update law in \eqref{ADMM_aux_Kalman0} cannot be implemented in a fully distributed manner. To enable distributed implementation, we define an auxiliary variable:
\begin{equation}
\label{transfomed_lambda}
    \tilde{\lambda}_{t,l} = \sqrt{\mathbb{L}} \lambda_{t,l}.
\end{equation}
Pre-multiplying the dual variable update in \eqref{ADMM_aux_Kalman0} by \( \sqrt{\mathbb{L}} \), the DKF algorithm for state estimation using distributed ADMM becomes:
\begin{equation}
\label{ADMM_aux_Kalman}
\begin{aligned}
    \tilde{\lambda}_{t,l+1} &= \tilde{\lambda}_{t,l} + \alpha_\lambda K_t^{-1} \mathbb{L} \xi_{t,l}, \\
    \xi_{t,l+1} &= K_t \big( \bar{\textbf{H}}^\top \bar{\textbf{S}}_t^{-1} \textbf{z}_t - \tilde{\lambda}_{t,l+1} - \mu K_t^{-1} \mathbb{L} \xi_{t,l} \big).
\end{aligned}
\end{equation}
The integral feedback term $\tilde{\lambda}_{t,l+1}$ in \eqref{ADMM_aux_Kalman} reduces the steady-state error in the consensus process.

Similarly, for the covariance matrix, we solve the distributed optimization problem given in \eqref{problem2}. To this end, we propose the following augmented Lagrangian:
\begin{equation}
\label{lag3}
L_{cov}(\theta, \nu) = \frac{1}{2}(N\omega^{\delta} - \theta)^\top (N\omega^{\delta} - \theta) + \nu^\top \sqrt{\tilde{\mathbb{L}}} \theta + \frac{\alpha_{\nu}}{2} \lVert \sqrt{\tilde{\mathbb{L}}} \theta \rVert^2,
\end{equation}
where \(\omega^{\delta}= [{\omega^{\delta}_1}; \cdots; {\omega^{\delta}_N} ]\), \( \tilde{\mathbb{L}} = \mathcal{L} \otimes I_{n_{\text{cov}}} \) and \( \alpha_{\nu} > 0 \) is a positive constant. 
To derive the update laws for \( \theta \) and \( \nu \), we take the gradient of the augmented Lagrangian in \eqref{lag3}, yielding:
\begin{equation}
\label{KKT3}
\begin{aligned}
    \nabla_{\theta} L &= -(N\omega^{\delta} - \theta) + \sqrt{\tilde{\mathbb{L}}} \nu + \alpha_{\nu} \tilde{\mathbb{L}} \theta, \\
    \nabla_{\nu} L &= \sqrt{\tilde{\mathbb{L}}} \theta.
\end{aligned}
\end{equation}
To enable a distributed update law for minimizing \eqref{lag3}, we introduce an auxiliary variable:
\begin{equation}
\tilde{\nu}_{t,l} = \sqrt{\tilde{\mathbb{L}}} \nu_{t,l}.
\end{equation}
The distributed update laws for the primal and dual variables are then given as:
\begin{equation}
\label{ADMM_covariance_update}
\begin{aligned}
    \tilde{\nu}_{t,l+1} &= \tilde{\nu}_{t,l} + \alpha_{\nu} \tilde{\mathbb{L}} \theta_{t,l}, \\
    \theta_{t,l+1} &= N\omega^{\delta} - \tilde{\nu}_{t,l+1} - \alpha_{\nu} \tilde{\mathbb{L}} \theta_{t,l}.
\end{aligned}
\end{equation}
Next, we parametrize the solution of \eqref{augumented_lag} using the saddle point equation.
\begin{lemma}
\label{lemma1}
    Let $P_{i, t \mid t-1}$ be a positive definite and symmetric matrix. Then the solution of \eqref{problem1}, considering the augmented Lagrangian \eqref{augumented_lag} is parameterized as $(\xi^*_t, \lambda^*_t)= ((1_N \otimes I_n)\xi^{\dag}_t,(1_N \otimes I_n) \hat{\lambda}_t+\bar{\lambda}_t)$, where $\xi^{\dag}_t = \mathcal H^{-1}_t \mathbb{1}^{\top}_N\bar{\mathbf{H}}^{\top}\bar{\mathbf{S}}_t^{-1}\mathbf{z}_t$. 
\end{lemma}
\begin{proof}
    Let $(\xi^*_t, \lambda^*_t)$ be the solution of \eqref{augumented_lag}. Using the Karush-Kuhn-Tucker (KKT) condition \cite[Theorem 12.1]{nocedal1999numerical}, \eqref{quad_aug_lag} can be written as:
    \begin{equation}
\label{KKT4}
    \begin{bmatrix}
\bar{\mathbf{H}}^{\top}\bar{\mathbf{S}}_t^{-1}\bar{\mathbf{H}}+ \mu \mathbb{L} & \sqrt{\mathbb{L}} \\
       \sqrt{\mathbb{L}} & \mathbb{O}
    \end{bmatrix}
    \begin{bmatrix}
        \xi_t^*\\
        \lambda_t^*
    \end{bmatrix}= \begin{bmatrix}
        \bar{\mathbf{H}}^{\top}\bar{\mathbf{S}}_t^{-1}\mathbf{z}_t \\
        0
    \end{bmatrix}.
\end{equation}
   From the primal feasibility condition in \eqref{KKT4}, we have:
    \[
    \sqrt{\mathbb{L}} \xi^*_t = 0,
    \]
    implying \( \xi^*_t \) lies in the nullspace of \( \sqrt{\mathbb{L}} \), which is spanned by \( 1_N \otimes I_n \). Thus, \( \xi^*_t = (1_N \otimes I_n) \xi^{\dag}_t \), where \( \xi^{\dag}_t \in \mathbb{R}^n \). 
% (\textcolor{red}{because $\sqrt{\mathbb{L}} \xi^*_t=0$, that means $\xi^*_t$ is in the nullspace of $\sqrt{\mathbb{L}}$.Without lost of generality, let $\xi^{\dag}_t \in \mathbb{R}^n$ be the first $n$ elements in $\xi^*_t$, then $\xi^*_t= \mathbb{1}_N \xi^{\dag}_t$}).
For the dual variable, the dual feasibility condition in \eqref{KKT4} can be written as:
\begin{equation}
\begin{aligned}
(\bar{\mathbf{H}}^{\top}\bar{\mathbf{S}}_t^{-1} & \bar{\mathbf{H}}+ \mu \mathbb{L})(1_N \otimes I_n)\xi^{\dag}_t+\sqrt{\mathbb{L}}\lambda^*_t = \bar{\mathbf{H}}^{\top}\bar{\mathbf{S}}_t^{-1}\mathbf{z}_t \\
    \sqrt{\mathbb{L}}\lambda^*_t &= \bar{\mathbf{H}}^{\top}\bar{\textbf{S}}_t^{-1}\mathbf{z}_t- (\bar{\mathbf{H}}^{\top}\bar{\mathbf{S}}_t^{-1}\bar{\mathbf{H}}+ \mu \mathbb{L})(1_N \otimes I_n)\xi^{\dag}_t \\
    \mathbb{L} \lambda^*_t &= \sqrt{\mathbb{L}}(\bar{\mathbf{H}}^{\top}\bar{\mathbf{S}}_t^{-1}\mathbf{z}_t- (\bar{\mathbf{H}}^{\top}\bar{\mathbf{S}}_t^{-1}\bar{\mathbf{H}}+ \mu \mathbb{L})(1_N \otimes I_n)\xi^{\dag}_t).
\end{aligned}
\end{equation}
Let $b = \sqrt{\mathbb{L}}(\bar{\mathbf{H}}^{\top}\bar{\mathbf{S}}_t^{-1}\mathbf{z}_t- (\bar{\mathbf{H}}^{\top}\bar{\mathbf{S}}_t^{-1}\bar{\mathbf{H}}+ \mu \mathbb{L})(1_N \otimes I_n)\xi^{\dag}_t)$. We know that $\mathcal L U = U \Lambda$ where $\Lambda = \diag(0, \lambda_2, \cdots, \lambda_N)$, $ U = \begin{bmatrix}
    u_N \; 
    \bar{U}
\end{bmatrix}$ with $u_N = \frac{1}{\sqrt{N}}1_N$, $1_N^{\top}\bar{U}= 0$, and $\bar{U}^{\top}\bar{U} = I_{N-1}$, we get 
$\lambda^*_t = U \otimes I_n [\hat{\lambda}_t; \bar{\Lambda}^{-1}\bar{U}^{\top} b]$.
\end{proof}
% \begin{remark}
%     \label{lambda_tilde}
%     The optimality in the transformed variable \eqref{transfomed_lambda} is achieved at zero, as $\tilde{\lambda}_t$ is initialized to zero.
% \end{remark}
Next, we parameterize the solution of \eqref{problem2} using the KKT condition.
\begin{lemma}\label{lemma2}
  The solution of \eqref{problem2} is parameterized as $(\theta^*_t, \nu^*_t)= ((1_N \otimes I_{n_{cov}})\theta^{\dag}_t,(1_N \otimes I_{n_{cov}}) \tilde{\nu}+\bar{\nu})$.     
\end{lemma}
\begin{proof}
     The proof of the Lemma~\ref{lemma2} follows the same path as the proof of Lemma~\ref{lemma1}, thus omitted.
%     Again, using the KKT condition, \eqref{KKT3} takes the following form:
%     \begin{equation}
% \label{KKT2}
%     \begin{bmatrix}
%        N \mathbb{I}+ \alpha_{\nu} \tilde{\mathbb{L}} & \sqrt{\tilde{\mathbb{L}}} \\
%        \sqrt{\tilde{\mathbb{L}}} & O
%     \end{bmatrix}
%     \begin{bmatrix}
%         \theta_t^*\\
%         \nu_t^*
%     \end{bmatrix}= \begin{bmatrix}
%         N\omega^{\delta}\\
%         O
%     \end{bmatrix}
% \end{equation}
% The primal feasibility equation in \eqref{KKT2} implies that $\theta^* = (1_N \otimes I_{n_{cov}})\theta^{\dag}$. For the optimal dual variable, the dual feasibility equation can be written as:
% \begin{equation}
% \begin{split}
%     (N \mathbb{I}+ \alpha_{\nu} \tilde{\mathbb{L}})&(1_N \otimes I_{n_{cov}})\theta^{\dag}_t +\sqrt{\tilde{\mathbb{L}}}\nu^* = N \omega^{\delta} \\
%     \sqrt{\tilde{\mathbb{L}}}\nu^* &= N\omega^{\delta} -(N \mathbb{I}+ \alpha_{\nu} \tilde{\mathbb{L}})(1_N \otimes I_{n_{cov}})\theta^{\dag}_t \\
%     \tilde{\mathbb{L}} \nu^* &= \sqrt{\tilde{\mathbb{L}}}(N\omega^{\delta} -(N \mathbb{I}+ \alpha_{\nu} \tilde{\mathbb{L}})(1_N \otimes I_{n_{cov}})\theta^{\dag}_t)
% \end{split}
% \end{equation}
% Let $b = \sqrt{\tilde{\mathbb{L}}}(N\omega^{\delta} -(N \mathbb{I}+ \alpha_{\nu} \tilde{\mathbb{L}})(1_N \otimes I_{n_{cov}})\theta^{\dag}_t)$. We know that $\mathcal L U = U \Lambda$ where $\Lambda = \diag(0, \lambda_2, \cdots, \lambda_N)$ and $ U = \begin{bmatrix}
%     1_N/\sqrt{N} \\ 
%     \bar{U}
% \end{bmatrix}$. With these, we have 
% $\nu^* = (U \otimes I_{n_{cov}} )[\tilde{\nu}; \bar{\Lambda} \bar{U}^{\top} b]$.
\end{proof}
A pseudo-code implementation of the proposed filtering method is provided in Algorithm \ref{Algo_ADMM2}. 
%==============================
 \begin{algorithm}[h!]
		\caption{DKF using ADMM}\label{Algo_ADMM2}
		\begin{algorithmic}[1]
  \Function{$[\hat{x}_{i, t|t}, \, P_{i, t|t}] = \text{DKF}$}{$\hat{x}_{i, 0\mid 0}, P_{i, 0\mid 0}$}. %, \Theta_k
  \For{$t = 1, \ldots, T$}
  \State Compute the prior mean at node $i$, %$ $
    \Statex \hspace{0.5 cm} $\hat{x}_{i, t\mid t-1}= F \hat{x}_{i, t-1 \mid t-1}.$
  % \begin{equation*}
  %     \xi_{i,t}=\hat{x}_{i, t\mid t-1} = F \hat{x}_{i, t-1 \mid t-1}.
  % \end{equation*}
  \State Evaluate the prior error covariance at node $i$, \Statex \hspace{0.5 cm} $P_{i, t \mid t-1} = F P_{i, t-1 \mid t-1} F^{\top}+ Q$.
  \State $\xi_{i,t,0}=\hat{x}_{i, t\mid t-1}$, $\tilde{\lambda}_{i, t,0}=0$.
  \For{$l = 0, \ldots, L-1$}
  \State  $\tilde{\lambda}_{i, t,l+1} = \tilde{\lambda}_{i, t, l} + \alpha_{\lambda}  K^{-1}_{i, t}\sum_{j \in \mathcal{N}_i} a_{ij} (\xi_{i, t, l} - \xi_{j, t, l})$,
  \Statex \hspace{0.75 cm}where $K^{-1}_{i,t}= H_i^\top R_i^{-1} H_i + \frac{1}{N}P_{i, t\mid t-1}^{-1}$.
\State $ \xi_{i,t, l+1} = K_{i, t} (H_i^\top R_i^{-1} y_{i, t} + \frac{1}{N}P_{i, t\mid t-1}^{-1} \hat{x}_{i, t\mid t-1} $
\Statex \hspace{1.9 cm}$- \tilde{\lambda}_{i, t,l+1} - \mu K^{-1}_{i,t}\sum_{j \in \mathcal{N}_i} a_{ij} (\xi_{i,t, l} - \xi_{j,t, l}) .$ % \mathbb{L}\xi_t
\EndFor
\State  $\nu_{i, t} = \nu_{i, t-1} + \alpha_{\nu, i}  \sum_{j \in \mathcal{N}_i} a_{ij} (\theta_{i, t} - \theta_{j, t})$
\State $\theta_{i, t} = N \omega^{\delta}_i - \nu_{i, t}  - \alpha_{\nu} \sum_{j \in \mathcal{N}_i} a_{ij} (\theta_{i, t} - \theta_{j, t})$.
\State  $\hat{x}_{i,t \mid t}= \xi_{i,t, L}$,  $P_{i, t\mid t} = (P_{i, t\mid t-1}^{-1} + \Theta_{i, t})^{-1}$, 
\Statex \hspace{0.5 cm}where  $\Theta_{i, t} = \text{vech}^{-1}(\theta_{i, t})$. 
\EndFor
\EndFunction
  \end{algorithmic}
  \end{algorithm} 
We notice that the update step for the state estimate in \eqref{ADMM_aux_Kalman} are independent from the update step of the information rate matrix given in \eqref{ADMM_covariance_update}. Thus, in the following, we first show the boundedness and convergence of the covariance matrix for all the local estimators. Thenceforth, we show the convergence of the state estimate for all local estimators. 
%=============Block diagram starts here
%The block diagram structure for the update of state estimate is given in Fig. \ref{fig:block_diagram_state}.
% \begin{figure}[h!]
%     \centering
%     \includegraphics[width=1.01\linewidth]{DKF_paper.pdf}
%     \caption{Block diagram structure of the update step for the state estimate.}
%     \label{fig:block_diagram_state}
% \end{figure}

%
% \begin{figure}[h!]
%     \centering
%     \includegraphics[width=1.01\linewidth]{DKF_paper.pdf}
%     \caption{Block diagram structure of the update step for the state estimate.}
%     \label{fig:block_diagram_state}
% \end{figure}
%================================================%
\section{Stability Analysis}
In this section, we analyze the stability of the proposed DKF algorithm. The update laws for the state estimate \eqref{ADMM_aux_Kalman} and the posterior covariance \eqref{ADMM_covariance_update} are modeled as discrete-time dynamical systems \cite{nishihara2015general}. Using tools from system theory, we derive conditions on the design parameters that guarantee the asymptotic stability of the update laws given in \eqref{ADMM_aux_Kalman} and \eqref{ADMM_covariance_update}. Finally, we show that all local estimators are unbiased, that is,  \( \lim_{t \to \infty} \mathbb{E}[x_t - \xi_{i,t}] = 0 \).
% In this section, we provide the stability analysis of our proposed DKF algorithm using consensus-based ADMM. First, we establish the boundedness and convergence of the local covariance matrices. The convergence results for the local covariance matrices are similar to those in \cite{ryu2023consensus}, and are therefore omitted. The update steps for the state estimate \eqref{ADMM_aux_Kalman} and the posterior covariance \eqref{ADMM_covariance_update} can be treated as discrete-time dynamical systems \cite{nishihara2015general}. We apply tools from system theory to derive conditions on the design parameters that ensure the asymptotic stability of \eqref{ADMM_aux_Kalman} and \eqref{ADMM_covariance_update}. Finally, we demonstrate the asymptotic stability of the entire DKF algorithm and show that  $\lim_{t \to \infty} \mathbb{E}\{x_t - \xi_{i,t}\}=0$.
\begin{theorem}
\label{theorem1_covariance}
Let the communication network \( \mathcal{G} \) of local estimators be undirected and connected, and let \( \mathcal{L} \) be the Laplacian matrix of \( \mathcal{G} \). If \( 0 < \alpha_{\nu} < \frac{2}{3 \lambda_{\text{max}}(\mathcal{L})} \), then the sequence \( \{\Theta_{i,t} = \mathrm{vech}^{-1}(\theta_{i,t})\}_t \) generated by \eqref{ADMM_covariance_update} converges to the global information rate matrix \( H^\top \bar{R}^{-1} H \).
\end{theorem}
\begin{proof}
The update law for \( \theta_{t,l} \) in \eqref{ADMM_covariance_update} can be written as:
\begin{equation}
    \label{theta_second}
    \theta_{t,l+1} = (\mathbb{I} - 2\alpha_{\nu} \tilde{\mathbb{L}}) \theta_{t,l} + \alpha_{\nu} \tilde{\mathbb{L}} \theta_{t,l-1}.
\end{equation}
To analyze the stability, we apply the coordinate transformation:
\begin{equation}
    \label{transformation_coordinate}
    \theta_{t,l} = \mathbb{U}^\top \phi_{t,l}, \quad \tilde{\nu}_{t,l} = \mathbb{U}^\top \psi_{t,l},
\end{equation}
where \( \mathbb{U} = U \otimes I_{n_{\text{cov}}} \).
Substituting \eqref{transformation_coordinate} into \eqref{theta_second}, we obtain:
\begin{equation}
    \label{transformed_dynamics}
    \begin{split}
        \psi_{t,l+1} &= \psi_{t,l} + \alpha_{\nu} \left( \Lambda \otimes I_{n_{cov}} \right) \phi_{t,l},\\
        \phi_{t,l+1} &= (I - 2\alpha_{\nu} (\Lambda \otimes I_{n_{\text{cov}}})) \phi_{t,l} + \alpha_{\nu} (\Lambda \otimes I_{n_{\text{cov}}}) \phi_{t,l-1}.
    \end{split}
\end{equation}
Let \( \psi_{t,l} = [\tilde{\psi}_{t,l}^\top, \; \bar{\psi}_{t,l}^\top]^\top \) and \( \phi_{t,l} = [\tilde{\phi}_{t,l}^\top, \; \bar{\phi}_{t,l}^\top]^\top \), where \( \tilde{\psi}_{t,l} \) and \( \tilde{\phi}_{t,l} \) correspond to the nullspace of \( \tilde{\mathbb{L}} \), whereas \( \bar{\psi}_{t,l} \) and \( \bar{\phi}_{t,l} \) correspond to the range space of \( \tilde{\mathbb{L}}\). The dynamics in \eqref{transformed_dynamics} decompose as:
\begin{equation}
    \label{decomposed_dynamics}
    \begin{aligned}
     \tilde{\psi}_{t, l+1} &= \tilde{\psi}_{t,l} \\
        \bar{\psi}_{t,l+1} &= \bar{\psi}_{t,l} + \alpha_{\nu} (\tilde{\Lambda} \otimes I_{n_{cov}}) \bar{\phi}_{t,l}, \\
        \tilde{\phi}_{t,l+1} &= \tilde{\phi}_{t,l}, \\
        \bar{\phi}_{t,l+1} &= (I - 2\alpha_{\nu} (\tilde{\Lambda} \otimes I_{n_{\text{cov}}})) \bar{\phi}_{t,l} + \alpha_{\nu} (\tilde{\Lambda} \otimes I_{n_{\text{cov}}}) \bar{\phi}_{t,l-1},
    \end{aligned}
\end{equation}
where \( \tilde{\Lambda} = \mathrm{diag}(\lambda_2, \cdots, \lambda_N) \).
Defining \( \hat{\phi}_{t,l} = [\bar{\phi}_{t,l}^\top, \bar{\phi}_{t,l-1}^\top]^\top \), we write the dynamics of \( \bar{\phi}_{t,l} \) as:
\begin{equation}
    \label{compact_dynamics}
    \hat{\phi}_{t,l+1} = M \hat{\phi}_{t,l},
\end{equation}
where
\[
M = \begin{bmatrix}
    I - 2\alpha_{\nu} (\tilde{\Lambda} \otimes I_{n_{\text{cov}}}) & \alpha_{\nu} (\tilde{\Lambda} \otimes I_{n_{\text{cov}}}) \\
    I & \mathbb{O}
\end{bmatrix}.
\]
The eigenvalues of \( M \) determine the stability of the system given in \eqref{compact_dynamics}. For each eigenvalue \( \lambda_i(\tilde{\Lambda}) \), we define \[ M_i=\begin{bmatrix}
    1 - 2\alpha_{\nu} \lambda_i(\tilde{\Lambda})  & \alpha_{\nu} \lambda_i(\tilde{\Lambda})  \\
    1 & 0
\end{bmatrix}. \]
The eigenvalues of $M_i$ can be written as:
\[
\lambda(M_i) = \frac{1}{2} \left( 1 - 2\alpha_{\nu} \lambda_i \pm \sqrt{(1 - 2\alpha_{\nu} \lambda_i)^2 + 4\alpha_{\nu} \lambda_i} \right).
\]
For the asymptotic stability of \eqref{compact_dynamics}, \( |\lambda(M_i)| < 1 \). This condition holds if \( 0 < \alpha_{\nu} < \frac{2}{3 \lambda_{\text{max}}(\mathcal{L})} \).
To ensure the stability of the entire system, we analyze the boundedness of \( \bar{\psi}_{t,l} \). From \eqref{decomposed_dynamics}, the dynamics of \( \bar{\psi}_{t,l} \) are given by:
\begin{equation}
    \label{stability_psi}
    \bar{\psi}_{t,l+1} = \sum_{k=0}^{l} \alpha_{\nu} (\tilde{\Lambda} \otimes I_{n_{\text{cov}}}) \bar{\phi}_{t,k}.
\end{equation}
Since \( \bar{\phi}_{t,l} \to 0 \) exponentially as \( l \to \infty \), there exist constants \( C < \infty \) and \( 0 < \rho < 1 \) such that:
\begin{equation}
    \label{phi_decay}
    \| \bar{\phi}_{t,l} \| \leq C \rho^l.
\end{equation}
Using \eqref{phi_decay}, the norm of \( \bar{\psi}_{t,l} \) is bounded as:
\begin{equation}
    \label{psi_bound}
    \| \bar{\psi}_{t,l} \| \leq \alpha_{\nu} \| \tilde{\Lambda} \otimes I_{n_{\text{cov}}} \| \frac{C}{1 - \rho}.
\end{equation}
Thus, \( \bar{\psi}_{t,l} \) remains bounded.
Since \( \tilde{\phi}_{t,l} \) corresponds to the nullspace of \( \tilde{\mathbb{L}} \), it remains constant. Specifically, we have:
\[
\tilde{\phi}_{t,\infty} = (u_N^\top \otimes I_{n_{\text{cov}}}) N \omega^\delta = \sqrt{N}\sum^N_{i=1}\mathrm{vech}(H^\top_i {R}^{-1}_i H_i).
\]
Thus, under the condition \( 0 < \alpha_{\nu} < \frac{2}{3 \lambda_{\text{max}}(\mathcal{L})} \), \( \{\Theta_{i,t}\}_t \) converges to \( H^\top \bar{R}^{-1} H \), completing the proof.
\end{proof}
Notice that the distributed optimization problem in \eqref{problem2} is static. Consequently, the update law in \eqref{ADMM_covariance_update} can be expressed without sub-iterations as (as in Algorithm~\ref{Algo_ADMM2}):
\begin{equation}
    \label{ADMM_covariance_update_noiteration}
    \begin{split}
        \tilde{\nu}_{t+1} &= \tilde{\nu}_{t} + \alpha_{\nu} \tilde{\mathbb{L}} \theta_{t}, \\
        \theta_{t+1} &= N\omega^{\delta} - \tilde{\nu}_{t+1} - \alpha_{\nu} \tilde{\mathbb{L}} \theta_{t}.
    \end{split}
\end{equation}
The sequences generated by \eqref{ADMM_covariance_update_noiteration} converge under the same conditions stated in Theorem~\ref{theorem1_covariance}, which we summarize in the following corollary.
\begin{corollary}
    Let the communication network \( \mathcal{G} \) of local estimators be undirected and connected, and let \( \mathcal{L} \) be the Laplacian matrix of \( \mathcal{G} \). If \( 0 < \alpha_{\nu} < \frac{2}{3 \lambda_{\text{max}}(\mathcal{L})} \), then the sequence \( \{\Theta_{i,t} = \mathrm{vech}^{-1}(\theta_{i,t})\}_t \) generated by \eqref{ADMM_covariance_update_noiteration} converges to the global information rate matrix \( H^\top \bar{R}^{-1} H \).
\end{corollary}
\begin{proof}
We begin by applying the same coordinate transformation as in Theorem~\ref{theorem1_covariance}:
\begin{equation}
    \label{transformation_coordinate2}
    \theta_{t} = \mathbb{U}^\top \phi_{t}, \quad \tilde{\nu}_{t} = \mathbb{U}^\top \psi_{t}.
\end{equation}
Substituting \eqref{transformation_coordinate2} into \eqref{ADMM_covariance_update_noiteration}, the transformed dynamics are given by:
\begin{equation}
\begin{split}
    \psi_{t+1} &= \psi_{t} + \alpha_{\nu} (\Lambda \otimes I_{n_{\text{cov}}}) \phi_{t}, \\
    \phi_{t+1} &= (I - 2\alpha_{\nu} (\Lambda \otimes I_{n_{\text{cov}}})) \phi_{t} + \alpha_{\nu} (\Lambda \otimes I_{n_{\text{cov}}}) \phi_{t-1}.
\end{split}
\end{equation}
The above dynamics are identical to those analyzed in Theorem~\ref{theorem1_covariance}, except that the update law in \eqref{ADMM_covariance_update_noiteration} does not include sub-iterations. Since the distributed optimization problem is static, the stability and convergence analysis in Theorem~\ref{theorem1_covariance} directly apply. Therefore, the details are omitted here for brevity.
\end{proof}
Next, we show the boundedness of the posterior covariance matrix $P_{i,t}$ and the prior covariance matrix $P_{i,t|t-1}$. 
\begin{lemma}
    Consider the ADMM algorithm for DKF, as described in Algorithm 1. Assume that Assumptions 1 and 2 hold. If $0 < \alpha_{\nu} < \frac{2}{3 \lambda_{\text{max}}(\mathcal L)}$, then there exist positive symmetric matrices $\bar{P}< \infty$ and $\underbar{P}> 0$ such that $\underbar{P}<P_{i,t}<\bar{P}< $ and $\underbar{P}<P_{i,t|t-1}<\bar{P}$.
\end{lemma}
\begin{proof}
    The proof follows the same path as given in Lemma 8 in \cite{ryu2023consensus}, thus omitted.
\end{proof}
Next, we show that the local covariance matrices $P_{i,t|t-1}$ for all $i \in \{1,2, \ldots, N\}$ generated by Algorithm 1 converge to $P^*$, which is a unique solution of the following algebraic Riccati equation:
\begin{equation}
\label{riccati_algebraic}
P = F \{P - P H^{\top}(H^{\top}PH+ \bar{R})^{-1}HP\}F^{\top}+Q.
\end{equation}
\begin{theorem}
    Consider Algorithm 1 and let Assumptions 1 and 2 hold. Let $P^*$ be the unique solution of \eqref{riccati_algebraic}. Let $0 < \alpha_{\nu} < \frac{2}{3 \lambda_{max}(\mathcal L)}$. Then the covariance matrices $P_{i,t|t-1}$ for all $i \in \{1,2, \ldots, N\}$ generated by Algorithm 1 converge to $P^*$.
\end{theorem}
\begin{proof}
       The proof follows the same path as that of Theorem 9 in \cite{ryu2023consensus}, thus omitted.
\end{proof}
Next, we establish the stability of the sequence \( \{\xi_{t,l}\} \) generated by Algorithm 1. To this end, we reformulate the dynamics of \( \xi_{t,l} \) independently of \( \tilde{\lambda}_{t,l} \). This reformulation allows us to use stability tools from dynamical systems theory, providing bounds on the design parameters \( \alpha_{\lambda} \) and \( \mu \). Specifically, we define the error terms as:
\begin{equation}
    \label{error1}
    e^{\xi}_{t,l} = \xi^*_t - \xi_{t,l}, \quad e^{\tilde{\lambda}}_{t,l} = \tilde{\lambda}^*_t - \tilde{\lambda}_{t,l}.
\end{equation}
The update step for \( \xi_{t,l} \) in \eqref{ADMM_aux_Kalman} can be expressed as a second-order discrete-time dynamical system:
\begin{equation}
    \label{correction_Alg2_2}
    \xi_{t,l+1} = \xi_{t,l} + \mu \mathbb{L}\xi_{t,l-1} - (\alpha_{\lambda} + \mu) \mathbb{L}\xi_{t,l}.
\end{equation}
Substituting the error definitions from \eqref{error1} into \eqref{ADMM_aux_Kalman} and \eqref{correction_Alg2_2}, the error dynamics are given by:
\begin{equation}
    \label{correction_Alg2_3}
    \begin{split}
        e^{\tilde{\lambda}}_{t,l+1} &= e^{\tilde{\lambda}}_{t,l} + \alpha_{\lambda} K^{-1}_t \mathbb{L} e^{\xi}_{t,l}, \\
        e^{\xi}_{t,l+1} &= e^{\xi}_{t,l} - (\alpha_{\lambda} + \mu) \mathbb{L} e^{\xi}_{t,l} + \mu \mathbb{L} e^{\xi}_{t,l-1}.
    \end{split}
\end{equation}
The next result establishes conditions on the design parameters \( \alpha_{\lambda} \) and \( \mu \) in terms of \( \lambda_{\text{max}}(\mathcal{L}) \). In contrast to \cite{ryu2023consensus}, where the parameters are related to \( \lambda^2_{\text{max}}(\mathcal{L}) \), the bounds provided here mitigate the slowdown of the consensus process, reducing the need for a larger \( L \) to achieve improved performance.

\begin{theorem}
\label{theorem5}
Consider the discrete-time dynamical system given in \eqref{correction_Alg2_3} with \( \alpha_{\lambda} > 0 \) and \( \mu > 0 \) such that \( \alpha_{\lambda} + 2\mu < \frac{2}{\lambda_{\text{max}}(\mathcal{L})} \). Let Assumptions 1 and 2 hold. Then the state vector \( e^{\xi}_{t,l} \) converges to zero asymptotically for sufficiently large \( L \), and \( e^{\tilde{\lambda}}_{t,l} \) remains bounded for any \( l \). Consequently, \( \xi_{t,l} \) in \eqref{ADMM_aux_Kalman} asymptotically converges to \( \xi^*_t \).
\end{theorem}
\begin{proof}
To analyze the stability of the equilibrium of \eqref{correction_Alg2_3}, we introduce the following coordinate transformation:
\begin{equation}
    \label{th5_trans}
    {e}^{\xi_T}_{t,l} = \mathbb{\bar{U}}^\top {e}^{\xi}_{t,l},
\end{equation}
where the subscript \( T \) denotes the transformed variable, and \( \mathbb{\bar{U}} = \bar{U} \otimes I_n \). Substituting \eqref{correction_Alg2_3} into \eqref{th5_trans}, the discrete-time dynamical system can be written in the transformed coordinates as:
\begin{equation}
    \label{autonomous_theorem3}
    {e}^{\xi_T}_{t,l+1} = e^{\xi_T}_{t,l} - (\alpha_{\lambda} + \mu) \mathbb{\tilde{\Lambda}} e^{\xi_T}_{t,l} + \mu \mathbb{\tilde{\Lambda}} e^{\xi_T}_{t,l-1},
\end{equation}
where \( \mathbb{\tilde{\Lambda}} =  \tilde{\Lambda} \otimes I_n \) with
\( \tilde{\Lambda} = \mathrm{diag}(\lambda_2(\mathcal{L}), \ldots, \lambda_N(\mathcal{L})) \).
Letting \( \bar{e}^{\xi_T}_{t,l} = [{e}^{{\xi_T}^\top}_{t,l}, {e}^{{\xi_T}^\top}_{t,l-1}]^\top \), we rewrite \eqref{autonomous_theorem3} in compact form:
\begin{equation}
    \label{correction_Alg2_6}
    \bar{e}^{\xi_T}_{t,l+1} = \bar{M} \bar{e}^{\xi_T}_{t,l},
\end{equation}
where
\[
\bar{M} = \begin{bmatrix}
    \mathbb{I} - (\alpha_{\lambda} + \mu) \mathbb{\tilde{\Lambda}} & \mu \mathbb{\tilde{\Lambda}} \\
    \mathbb{I} & \mathbb{O}
\end{bmatrix}.
\]
To show the stability of \eqref{correction_Alg2_6}, we ensure that \( \bar{M} \) is Schur stable. The eigenvalues of \( \bar{M} \) are determined by the eigenvalues of the following \( 2 \times 2 \) matrix for each eigenvalue \( \lambda_i(\tilde{\Lambda}) \) of \( \tilde{\Lambda} \):
\[
\bar{M}^i = \begin{bmatrix}
    \bar{\alpha}_{\lambda \mu} & \mu \lambda_i \\
    1 & 0
\end{bmatrix}, \,\, \text{where} \,\, \bar{\alpha}_{\lambda \mu} = 1 - (\alpha_{\lambda} + \mu) \lambda_i.
\]
The eigenvalues of \( \bar{M}^i \) are given by:
% \begin{equation*}
%     \begin{split}
        
%     \end{split}
% \end{equation*}
\[
\lambda(\bar{M}^i) = \frac{1}{2} \left(\bar{\alpha}_{\lambda \mu} \pm \sqrt{\bar{\alpha}_{\lambda \mu}^2 + 4\mu \lambda_i} \right).
\]
For \( \bar{M}^i \) to be Schur stable, we require \( |\lambda(\bar{M}^i)| < 1 \). This condition is satisfied if \( \alpha_{\lambda} > 0 \), \( \mu > 0 \), and \( \alpha_{\lambda} + 2\mu < \frac{2}{\lambda_{\text{max}}(\mathcal{L})} \) for all \( i \in \{1, \ldots, N - 1\} \). Under these conditions, \( \bar{M} \) is Schur stable, and \( \bar{e}^{\xi_T}_{t,l} \to 0 \) as \( l \to \infty \). Consequently, \( e^{\xi}_{t,l} \to 0 \) as \( l \to \infty \) since \( {e}^{\xi}_{t,l} = \mathbb{\bar{U}} {e}^{\xi_T}_{t,l} \).
To establish the boundedness of \( e^{\tilde{\lambda}}_{t,l} \), consider its update equation in \eqref{correction_Alg2_3}:
\[
e^{\tilde{\lambda}}_{t,l+1} = e^{\tilde{\lambda}}_{t,l} + \alpha_{\lambda} K^{-1}_t \mathbb{L} e^{\xi}_{t,l}.
\]
Since \( e^{\xi}_{t,l} \to 0 \) as \( l \to \infty \), \( e^{\tilde{\lambda}}_{t,l} \) satisfies the following bound:
\begin{equation}
    \label{boundednes_lambdatilde}
    \| e^{\tilde{\lambda}}_{t,l} \| \leq \| e^{\tilde{\lambda}}_{t,0} \| + \alpha_{\lambda} \| K^{-1}_t \| \frac{C}{1 - \rho},
\end{equation}
where \( C < \infty \) and \( 0 < \rho < 1 \) are constants determined by the exponential decay of \( e^{\xi}_{t,l} \). Thus, \( e^{\tilde{\lambda}}_{t,l} \) remains bounded.
Under the conditions \( \alpha_{\lambda} > 0 \), \( \mu > 0 \), and \( \alpha_{\lambda} + 2\mu < \frac{2}{\lambda_{\text{max}}(\mathcal{L})} \), \( e^{\xi}_{t,l} \to 0 \) asymptotically, and \( e^{\tilde{\lambda}}_{t,l} \) remains bounded. Consequently, \( \xi_{t,l} \to \xi^*_t \) as \( l \to \infty \), completing the proof.
\end{proof}
Next, we show the stability of the state estimates as $t$ approach to infinity.
\begin{theorem}
\label{main_theorem_algo1}
    Let the Assumptions 1-2 hold. Then the sequence generated by Algorithm 1 satisfies the following:
    \begin{equation}
        \lim_{t \to \infty} \mathbb{E}\{x_t - \xi_{i,t}\}=0.
    \end{equation}
\end{theorem}
\begin{proof}
    Define $\mu^{\xi}_{t} = \mathbb{E} \{e^{\xi}_{t,l}\} =\mathbb{E} \{ \xi^*_t- \xi_{t,l}\}$, $\mu^{\tilde{\lambda}}_t = \mathbb{E} \{e^{\tilde{\lambda}}_t \}=\mathbb{E} \{ \tilde{\lambda}^*_t- \tilde{\lambda}_{t,l}\}$. Next, we define a function $V^{\dag}_{t} 
    =\mu^{\dag^{\top}}_{t} \mathcal H_t \mu^{\dag}_{t}$, where $\mu^{\dag}_t= \mathbb{E}\{x_t - \xi^{\dag}_t\}$, $\mathcal H_t = \mathbb{1}^{\top}_N \bar{\mathbf{H}}^{\top} \bar{\mathbf{S}}^{-1}_t \bar{\mathbf{H}}\mathbb{1}_N$, and $\xi^{\dag}_t = \mathcal H_t^{-1} \mathbb{1}^{\top}_N \bar{\mathbf{H}}^{\top} \bar{\mathbf{S}}^{-1}_t \mathbf{z}_t$.
    We can write
    \begin{equation}
        \label{lyap1_dag}
        \begin{split}
            V^{\dag}_{t+1}&= -\mu^{\dag^{\top}}_{t+1}(\mathcal H_{t+1} - 2 \mathcal H_{t+1})\mu^{\dag}_{t+1} \\
            & = -\mu^{\dag^{\top}}_{t+1} \mathbb{1}^{\top}_N \bar{{H}}^{\top} \bar{{R}}^{-1} \bar{{H}}\mathbb{1}_N\mu^{\dag}_{t+1}-\\
            &\frac{1}{N}\mu^{\dag^{\top}}_{t+1} \mathbb{1}^{\top}_N  \bar{{P}}^{-1}_{t+1|t}\mathbb{1}_N\mu^{\dag}_{t+1}+
            2\mu^{\dag^{\top}}_{t+1} \mathcal H_{t+1} \mu^{\dag}_{t+1}, 
        \end{split}
    \end{equation}
where $\bar{H}= \blockdiag\{H_1, \cdots, H_N \}$, $\bar{P}_{t+1|t}= \blockdiag\{P_{1,t+1|t}, \cdots, P_{N,t+1|t} \}$, and $\bar{R}= \blockdiag\{R_1, \cdots, R_N \}$.  Let $\mathcal P_{t+1|t}= (\frac{1}{N}\mathbb{1}^{\top}_N \bar{P}^{-1}_{t+1|t}\mathbb{1}_N )^{-1}$. To further expand \eqref{lyap1_dag}, we write the dynamics of $\mu^{\dag}_{t+1}$ in a convenient form. Using the definition of $\mu^{\dag}_{t+1}$, we have 
    \begin{equation*}
        \label{e_dagger2}
        \begin{aligned}
            \mu^{\dag}_{t+1} &= \mathbb{E}\{x_{t+1}-\xi^{\dag}_{t+1}\}\\
            &= \mathbb{E}\{x_{t+1}\}-\mathcal H^{-1}_{t+1}\mathbb{1}^{\top}_N \mathbb{E}\{\bar{\mathbf{H}}^{\top}\bar{\mathbf{S}}^{-1}_{t+1}\mathbf{z}_{t+1}\}\\
            &=\mathcal H^{-1}_{t+1} \mathcal H_{t+1} \mathbb{E}\{x_{t+1}\}-\mathcal H^{-1}_{t+1}\mathbb{1}^{\top}_N \mathbb{E}\{\bar{\mathbf{H}}^{\top}\bar{\mathbf{S}}_{t+1}\mathbf{z}_{t+1}\}\\
            & = \mathcal H^{-1}_{t+1}\mathbb{1}^{\top}_N\left(K^{-1}_{t+1} \mathbb{1}_N\mathbb{E}\{x_{t+1}\}-\mathbb{E}\{\bar{\mathbf{H}}^{\top}\bar{\mathbf{S}}_{t+1}\mathbf{z}_{t+1}\}\right),
        \end{aligned}
    \end{equation*}
    where $K^{-1}_{t+1} = \bar{{H}}^{\top}\bar{R}^{-1}\bar{{H}}  +\frac{1}{N}\bar{P}^{-1}_{t+1|t}$ and $\mathbb{E}\{\bar{\mathbf{H}}^{\top}\bar{\mathbf{S}}^{-1}_{t+1}\mathbf{z}_{t+1}\}= \bar{H}^{\top}\bar{R}^{-1}\bar{H} \mathbb{1}_N \mathbb{E}\{x_{t+1}\}+\frac{1}{N}\bar{P}^{-1}_{t+1|t}\mathbb{E}\{\hat{x}_{t+1|t}\}$. Next using $\hat{x}_{t+1|t} = \bar{F}\xi_{t,L}$, and that $\mathbb{1}_N F \mathbb{E}\{x_t\}-\bar{F}\mathbb{E}\{\xi_{t, L}\}+\mathbb{1}_N F\mathbb{E}\{\xi^{\dag}_t\}-\mathbb{1}_N F\mathbb{E}\{\xi^{\dag}_t\}=\mathbb{1}_N F\mu^{\dag}_t+\bar{F}\mu^{\xi}_{t} $, we have 
    \begin{equation}
    \label{e_dagger3}
        \begin{aligned}
           \mu^{\dag}_{t+1} &= \mathcal H^{-1}_{t+1}\mathbb{1}^{\top}_N (\bar{H}^{\top}\bar{R}^{-1}\bar{H} \mathbb{1}_N \mathbb{E}\{x_{t+1}\}+
           \frac{1}{N}\bar{P}^{-1}_{t+1|t}\mathbb{1}_NF\mathbb{E}\{x_{t}\} \\
           &-\bar{H}^{\top}\bar{R}^{-1}\bar{H} \mathbb{1}_N \mathbb{E}\{x_{t+1}\}-\frac{1}{N}\bar{P}^{-1}_{t+1|t}\bar{F}\mathbb{E}\{\xi_{t,L} \}). \\
           & = E^{\dag}_{t+1} \mu^{\dag}_t +\bar{E}^{\dag}_{t+1} \mu^{\xi}_t,
        \end{aligned}
    \end{equation}
    where $E^{\dag}_{t+1} = \frac{1}{N}\mathcal H^{-1}_{t+1} \mathbb{1}^{\top}_N\bar{P}^{-1}_{t+1|t}\mathbb{1}_N F$ and $\bar{E}^{\dag}_{t+1} = \frac{1}{N}\mathcal H^{-1}_{t+1} \mathbb{1}^{\top}_N\bar{P}^{-1}_{t+1|t}\bar{F}$. Pre-multiplying \eqref{e_dagger3} by $\mathcal P_{t+1|t} \mathcal H_{t+1}$, we get \begin{equation}
    \label{e_dagger}
    \mathcal P_{t+1|t} \mathcal H_{t+1}    \mu^{\dag}_{t+1} = F \mu^{\dag}_{t}+ {G}^{\xi}_{t+1}\mu^{\xi}_{t},
    \end{equation}
    where ${G}^{\xi}_{t+1}=-\mathcal P_{t+1|t}\mathbb{1}^{\top}_N\frac{1}{N}\bar{P}^{-1}_{t+1|t}\bar{F} $. We will use \eqref{e_dagger} to write the last term in \eqref{lyap1_dag}.
   Next we re-write $\mu^{\dag}_{t+1}$ as below:
\begin{equation}
\label{control_form1}
\begin{split}
    \mu^{\dag}_{t+1} &= F \mu^{\dag}_{t}+ (E^{\dag}_{t+1} -F)\mu^{\dag}_{t}+ \bar{E}^{\dag}_{t+1} \mu^{\xi}_t\\
    & = F \mu^{\dag}_{t}+u_{t+1},
\end{split}
\end{equation}
where $u_{t+1}=(E^{\dag}_{t+1} -F)\mu^{\dag}_{t}+ \bar{E}^{\dag}_{t+1} \mu^{\xi}_t$.
Using \eqref{e_dagger}, \eqref{e_dagger3}, and \eqref{control_form1}, \eqref{lyap1_dag} can be written as:
\begin{equation}
    \label{Vk_dagger_equation1}
    \begin{split}
        V^{\dag}_{t+1} &= -\mu^{\dag^{\top}}_{t+1} {H}^{\top}\bar{R}^{-1}{H} \mu^{\dag}_{t+1}-u^{\top}_{t+1} \mathcal P^{-1}_{t+1|t} u_{t+1}\\
        & +\mu^{\dag^{\top}}_{t}F^{\top}\mathcal P^{-1}_{t+1|t}F\mu^{\dag}_{t}+2\mu^{\dag^{\top}}_{t+1}\mathcal P^{-1}_{t+1|t}{G}^{\xi}_{t+1}\mu^{\xi}_t.
    \end{split}
\end{equation}
Recalling that \( \mathcal P^{-1}_{t+1|t} = \mathbb{1}^{\top}_N \frac{1}{N} \bar{P}^{-1}_{t+1|t} \mathbb{1}_N \), we write:
\begin{equation}
    \mu^{\dag^{\top}}_{t} F^{\top} \mathcal P^{-1}_{t+1|t} F \mu^{\dag}_{t} = \mu^{\dag^{\top}}_{t} F^{\top} \sum^N_{i=1} P^{-1}_{i,t+1|t} F \mu^{\dag}_{t}.
\end{equation}
Rewriting \( P_{i,t+1|t} = F P_{i,t} F^{\top} + Q \) as \( P_{i,t+1|t} = F P_{i,t} F^{\top} + I^{\top} Q I \), and using the matrix inversion lemma along with continuity arguments, we obtain:
\begin{equation}
    \mu^{\dag^{\top}}_{t} F^{\top} \mathcal P^{-1}_{t+1|t} F \mu^{\dag}_{t} \leq \frac{1}{N} \mu^{\dag^{\top}}_{t} \sum^N_{i=1} P^{-1}_{i,t} \mu^{\dag}_{t}.
\end{equation}
For \( t \geq 1 \), note that \( P^{-1}_{i,t} = P^{-1}_{i,t|t-1} + \Theta_{i,t} \), and since \( \sum^N_{i=1} \Theta_{i,t} = N \sum^N_{i=1} H^{\top}_i R^{-1}_i H_i \), we conclude:
\begin{equation}
\label{inq_Vt}
    \mu^{\dag^{\top}}_{t} F^{\top} \mathcal P^{-1}_{t+1|t} F \mu^{\dag}_{t} \leq V^{\dag}_t.
\end{equation}
Next, let \( \bar{f} > 0 \) such that \( \lVert F \rVert_F \leq \bar{f} \), and using Lemma 4, there exist constants \( c_o \) and \( c_1 \), which depend on \( \lVert \bar{P} \rVert_F \), \( \lVert \underbar{P} \rVert_F \), and \( \bar{f} \), such that:
\begin{equation}
    \label{inequality4}
    \begin{aligned}
        2 \mu^{\dag^{\top}}_{t+1} \mathcal P^{-1}_{t+1|t} G^{\xi}_{t+1} \mu^{\xi}_t 
        &= 2 (E^{\dag}_{t+1} \mu^{\dag}_{t} + \bar{E}^{\dag}_{t+1} \mu^{\xi}_t)^{\top} \mathcal P^{-1}_{t+1|t} G^{\xi}_{t+1} \mu^{\xi}_t \\
        &\leq c_o \lVert \mu^{\dag}_{t} \rVert \lVert \mu^{\xi}_t \rVert + c_1 \lVert \mu^{\xi}_t \rVert^2.
    \end{aligned}
\end{equation}
Using \eqref{inequality4} and \eqref{inq_Vt} in \eqref{Vk_dagger_equation1}, we get
\begin{equation}
\label{lyap_diff_last}
\begin{aligned}
   V^{\dag}_{t+1} - V^{\dag}_t &\leq -\mu^{\dag^{\top}}_{t+1} {H}^{\top}\bar{R}^{-1}{H} \mu^{\dag}_{t+1}-u^{\top}_{t+1} \mathcal P^{-1}_{t+1|t} u_{t+1} \\
   &+c_o \lVert \mu^{\dag}_{t} \rVert \lVert \mu^{\xi}_t \rVert + c_1 \lVert \mu^{\xi}_t \rVert^2.
   \end{aligned}
\end{equation}
Define the following:
\begin{equation*}
    J_t = \sum^{T-1}_{s=0} \mu^{\dag^{\top}}_{t+1+s} H^{\top} \bar{R}^{-1} H \mu^{\dag}_{t+1+s} + u^{\top}_{t+1+s} \mathcal P^{-1}_{t+1+s|t+s} u_{t+1+s}.
\end{equation*}
Summing \eqref{lyap_diff_last} from \( t \) to \( t+T \), we obtain:
\begin{equation}
    \label{V_gad_difference}
    V^{\dag}_{t+T} - V^{\dag}_{t} \leq -J_t + \frac{c_o}{2\epsilon} \sum^{T-1}_{s=0} \lVert \mu^{\dag}_{t+s} \rVert^2 + \left( \frac{c_o \epsilon}{2} + c_1 \right) \sum^{T-1}_{s=0} \lVert \mu^{\xi}_{t+s} \rVert^2,
\end{equation}
where \( \epsilon > 0 \) is a constant. In \eqref{V_gad_difference}, we use Young's inequality.
Next, let \( \boldsymbol{\mathsf{\mu}}^{\dag}_{t+1} = [\mu^{\dag}_{t+1}; \cdots; \mu^{\dag}_{t+T}] \) and \( \boldsymbol{\mathsf{u}}_{t+1} = [u_{t+1}; \cdots; u_{t+T}] \). Rewriting \eqref{control_form1} compactly:
\begin{equation}
    \label{control_form2}
    \begin{aligned}
        \boldsymbol{\mathsf{\mu}}^{\dag}_{t+1} &= 
        \begin{bmatrix}
            F \\
            F^2 \\
            \vdots \\
            F^T
        \end{bmatrix} \mu^{\dag}_{t} + 
        \begin{bmatrix}
            I \\
            F & I \\
            \vdots & \ddots & \ddots \\
            F^{T-1} & \cdots & F & I
        \end{bmatrix} \boldsymbol{\mathsf{u}}_{t+1} \\
        &= \mathcal F \mu^{\dag}_{t} + \mathcal G \boldsymbol{\mathsf{u}}_{t+1}.
    \end{aligned}
\end{equation}
The function \( J_t \) can then be expressed as:
\begin{equation}
    \begin{aligned}
        J_t &= (\mathcal F \mu^{\dag}_t + \mathcal G \boldsymbol{\mathsf{u}}_{t+1})^{\top} \mathcal H^{\top} \mathcal{R}^{-1} \mathcal H (\mathcal F \mu^{\dag}_t + \mathcal G \boldsymbol{\mathsf{u}}_{t+1}) \\
        &\quad + \boldsymbol{\mathsf{u}}^{\top}_{t+1} \mathcal P^{-1}_{t+1|t} \boldsymbol{\mathsf{u}}_{t+1}.
    \end{aligned}
\end{equation}
This function is convex and quadratic in \( \boldsymbol{\mathsf{u}}_{t+1} \). Minimizing \( J_t \) with respect to \( \boldsymbol{\mathsf{u}}_{t+1} \), we obtain:
\begin{equation}
    J_t^* = \mu^{\dag^{\top}}_{t} \mathcal O^{\top} (\mathcal R + \mathcal H \mathcal G \mathcal P^{-1}_{t+1|t} \mathcal G^{\top} \mathcal H^{\top})^{-1} \mathcal O \mu^{\dag^{\top}}_{t},
\end{equation}
where \( \mathcal O \) is the observability matrix of the pair \( (F, H) \). Since \( \mathcal O \) is full-rank and \( \mathcal R > 0 \), there exists \( c^{\dag} > 0 \) such that:
\begin{equation}
    2 c^{\dag} \lVert \mu^{\dag}_t \rVert^2 \leq J_t^* \leq J_t.
\end{equation}
Using the relationship \( \lVert \mu^{\dag}_{t+1} \rVert^2 \leq \bar{c} (\lVert \mu^{\dag}_t \rVert^2 + \lVert \mu^{\xi}_t \rVert^2) \), with \( \tilde{c} = \max\{1, \bar{c}^{T-1}\} \), the summation term \( \sum_{s=0}^{T-1} \lVert \mu^{\dag}_{t+s} \rVert^2 \) satisfies:
\begin{equation}
    \label{second_term_65}
    \sum_{s=0}^{T-1} \lVert \mu^{\dag}_{t+s} \rVert^2 \leq (T-1) \tilde{c} \{ \lVert \mu^{\dag}_t \rVert^2 + \lVert \mu^{\xi}_t \rVert^2  \}.
\end{equation}
Similarly:
\begin{equation}
    \label{last_term_65}
    \sum_{s=0}^{T-1} \lVert \mu^{\xi}_{t+s} \rVert^2 \leq (T-1) \tilde{c} \{ \lVert \mu^{\xi}_t \rVert^2 \}.
\end{equation}
Using \eqref{second_term_65} and \eqref{last_term_65}, the bound in \eqref{V_gad_difference} becomes:
\begin{equation}
\label{V_gad_difference2}
\begin{aligned}
    V^{\dag}_{t+T} &- V^{\dag}_{t} \leq -2c^{\dag} \lVert \mu^{\dag}_{t} \rVert^2 
    + \frac{c_o }{2\epsilon}(T-1)\tilde{c} \{ \lVert \mu^{\dag}_{t} \rVert^2 
    + \lVert \mu^{\xi}_{t} \rVert^2 \}
    + 
    \\
    &\quad    + (\frac{c_o\epsilon}{2} + c_1)(T-1) \tilde{c} 
    \lVert \mu^{\xi}_{t} \rVert^2 .
\end{aligned}
\end{equation}
Letting $\epsilon = \frac{c_o (T-1)\tilde{c}}{2 c^{\dag}} $ and $c_2= (\frac{c_o\epsilon}{2}+c_1) (T-1) \tilde{c} +1$, we have 
\begin{equation}
    \label{V_gad_difference3}
    V^{\dag}_{t+T} - V^{\dag}_{t} \leq -c^{\dag} \lVert \mu^{\dag}_{t} \rVert^2 +   c_2  \lVert \mu^{\xi}_{t} \rVert^2 .
\end{equation}

Defining $\eta^{\hat{e}}_{t}= \mathbb{E}\{\bar{e}^{\xi_T}_{t,l}\}$ and taking the expectation of \eqref{correction_Alg2_6}, we write
\begin{equation}
    \label{compacr_error_exp}
    \eta^{\hat{e}}_{t+1} = \bar{M}\eta^{\hat{e}}_{t}.
\end{equation}
Next, we define ${V}^{\eta}_t = \eta^{\hat{e}^{\top}}_{t} \mathbb{P}\eta^{\hat{e}}_{t}$. Using the results in Theorem 3, we conclude that there exists a positive definite matrix $\Gamma$ such that
\begin{equation}
     \label{second_Lyap}
    {V}^{\eta}_{t+1}-{V}^{\eta}_t = -\eta^{\hat{e}^{\top}}_{t} \Gamma \eta^{\hat{e}}_{t}.
\end{equation}
%The zero eigenvalues of $\bar{M}$, and consequently of $\Gamma$, are due to the simple zero eigenvalue of $\mathcal L$, which corresponds to the condition where $\eta^{\hat{e}}_{t} = 0$. 
Using the transformation \eqref{th5_trans}, the following inequality can be derived from \eqref{second_Lyap}:
\begin{equation}
\label{second_Lyap_final}
    % {V}^{\eta}_{t+T}-{V}^{\eta}_t \leq -\lambda_{\mathbb{R}_{>0}}(\Gamma) \{ \lVert \mu^{\xi}_{t} \rVert^2+  \lVert \mu^{\tilde{\lambda}}_{t}\rVert^2 \},
     {V}^{\eta}_{t+T}-{V}^{\eta}_t \leq -c_3\lambda_{\mathbb{R}_{>0}}(\Gamma)  \lVert \mu^{\xi}_{t} \rVert^2
\end{equation}  
where $\lambda_{\mathbb{R}_{>0}}(\Gamma) > 0$ is the smallest positive eigenvalue of $\Gamma$, and $c_3>0$. Next, consider the following Lyapunov function candidate:
\begin{equation}
\label{final_lyap}
    V_t = {V}^{\dag}_t + \gamma {V}^{\eta}_t.
\end{equation}
Using \eqref{V_gad_difference3}, \eqref{second_Lyap_final}, and \eqref{final_lyap}, we obtain
\begin{equation*}
\begin{aligned}
      V_{t+T} - V_t &= {V}^{\dag}_{t+T} - {V}^{\dag}_t + \gamma ({V}^{\eta}_{t+T}-{V}^{\eta}_t) \\
      & \leq -c^{\dag} \lVert \mu^{\dag}_{t} \rVert^2 +   c_2 \lVert \mu^{\xi}_{t} \rVert^2 
      + \gamma (-\lambda_{\mathbb{R}_{>0}}(\Gamma) c_3 \lVert \mu^{\xi}_{t} \rVert^2 \\
      & \leq -c^{\dag} \lVert \mu^{\dag}_{t} \rVert^2 -(\gamma \lambda_{\mathbb{R}_{>0}}(\Gamma)c_3 -c_2) \lVert \mu^{\xi}_{t} \rVert^2 .
\end{aligned}
\end{equation*}
We choose $\gamma > 0$ and $\Gamma$ such that $\gamma \lambda_{\mathbb{R}_{>0}}(\Gamma)c_3 - c_2 > 0$, and the result follows.
\end{proof}

% \textcolor{blue}{Next, we show that the state estimate using Algorithm \ref{Algo_ADMM2} approaches to the state of the discrete-time dynamical system in \eqref{equ:system}.
% \begin{theorem}
% \label{main_theorem_algo2}
%     Let the Assumptions 1-2 hold.  If $\alpha_{\lambda} + 2\mu <\frac{2}{ \lambda_{\text{max}}(\mathcal L)}$, then the sequence generated by Algorithm 2 satisfies the following:
%     \begin{equation}
%         \lim_{t \to \infty} \mathbb{E}\{x_t - \xi_{i,t}\}=0
%     \end{equation}
% \end{theorem}
% \begin{proof}
%     The proof follows the same path as the proof of Theorem \ref{main_theorem_algo1}. Thus omitted.
% \end{proof}}
%\newpage 
\section{Simulation Results}
In this section, we validate the theoretical results of the proposed distributed filtering algorithm by simulating a network of 100 sensor nodes tracking the trajectory of a car moving with constant velocity, as described in \cite[pp.~99--101]{sarkka2023bayesian}. The dynamical system is collectively observable, with the state transition matrix given by  
\[
F = \begin{bmatrix}
    I_{2} & \delta t \, I_{2} \\
    O & I_{2}
\end{bmatrix},
\]
and the state vector is \( x = [x_1, x_2, x_3, x_4]^\top \). Here, \((x_1, x_2)\) represent the position of the car in the \(x-y\) plane, and \((x_3, x_4)\) represent the corresponding velocities. The measurement at the \(i\)-th sensor node observes either \(x_1\) or \(x_2\) randomly at each time step. The process noise covariance and measurement noise covariance are selected as per \cite[pp.~99--101]{sarkka2023bayesian}.

For the simulation, the following parameters are used: time step \(\delta t = 0.1\) seconds, and the design parameters are  \(\alpha_\lambda = 0.10\), \(\alpha_\nu = 0.04\), \(\mu = 0.001\). Each estimator at the \(i\)-th node is initialized with randomly selected \(\hat{x}_{i, 0\mid 0}\) and \(P_{i, 0\mid 0}\). The simulation spans a total duration of 10 seconds. To estimate the state, the proposed distributed filtering algorithm is implemented with 20 sub-iterations per time step, i.e., \(L = 20\). 

The filtering performance is evaluated in terms of the root mean squared error (RMSE) of the position and velocity estimates. In Fig.~\ref{fig:car_tracking_iteration}, we plot the RMSE of the position and velocity estimates for all estimators, i.e., the DKF at all sensor nodes, averaged over 50 Monte Carlo (MC) runs. Fig.~\ref{fig:car_tracking_iteration} demonstrates that the position and velocity RMSE values of the distributed estimators converge and achieve similar performance across all nodes, thereby validating the theoretical results presented in this paper.

\begin{figure}[h!]
		\centering
		\begin{subfigure}[b]{0.24\textwidth}
			\centering
			\includegraphics[width=4.8cm,height=3.5cm]{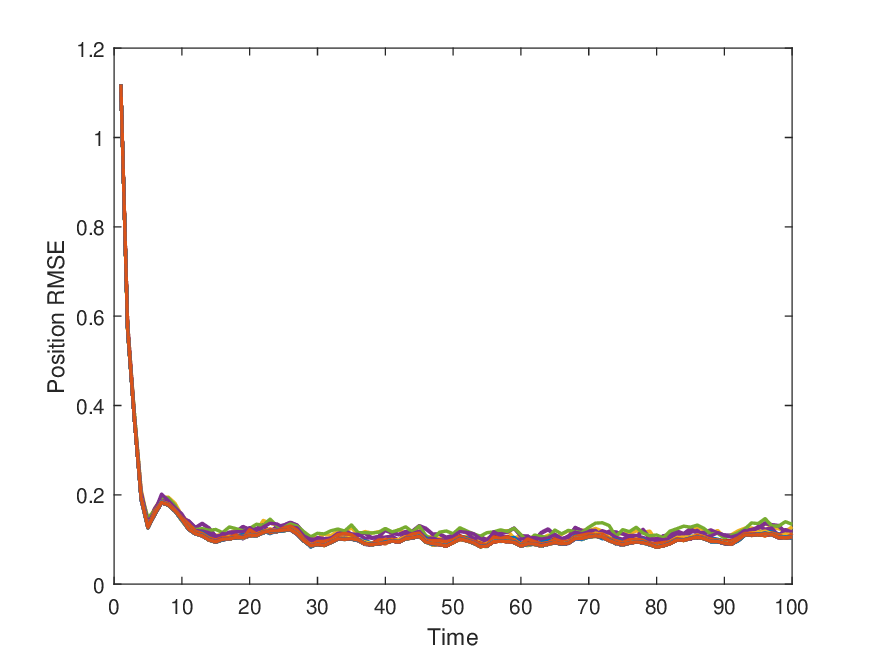}
			\caption{}
			\label{fig_rms_pos}
		\end{subfigure}
		\hfill
		\begin{subfigure}[b]{0.24\textwidth}
			\centering
			\includegraphics[width=4.8cm,height=3.5cm]{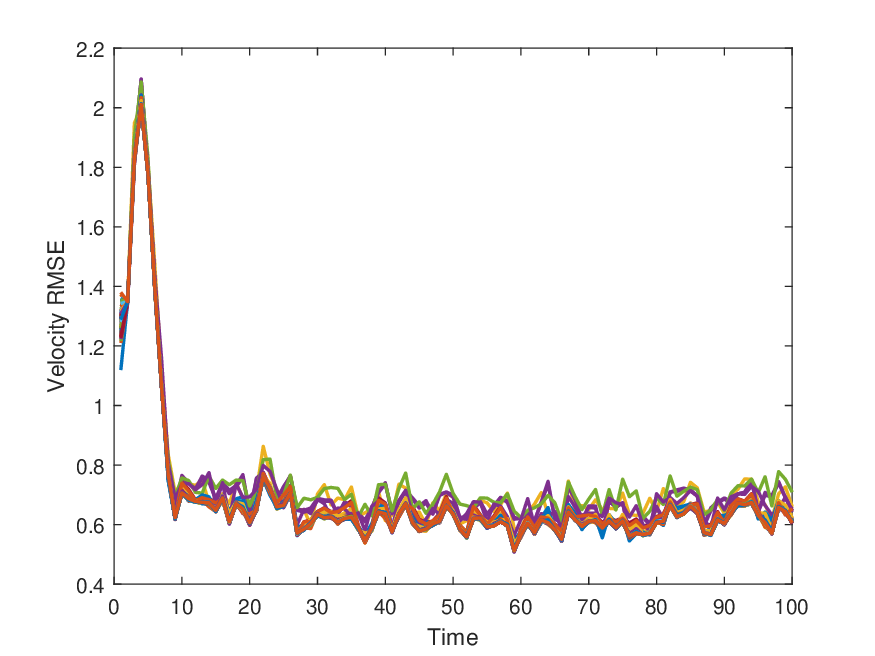}
			\caption{}
			\label{fig_rms_vel}
		\end{subfigure}		
		\caption{Position and velocity RMSE of the proposed DKF with different sub-iterations obtained from 50 MC runs.}
		\label{fig:car_tracking_iteration}
\end{figure}
\section{Conclusion}

We developed a consensus-based ADMM algorithm to derive the correction step for the distributed filtering. A new augmented Lagrangian formulation for the DKF problem was proposed, enabling a fully distributed implementation of the correction step for the posterior state and covariance estimates. The proposed consensus-based ADMM avoids exchanging dual variables, significantly reducing communication between nodes. The algorithm yields much tighter upper bounds, specifically $\alpha_{\nu} < \frac{2}{3 \lambda_{\text{max}}(\mathcal L)}$ and $\alpha_{\lambda}+2\mu <\frac{2}{ \lambda_{\text{max}}(\mathcal L)}$. 
Larger design parameter values improve the convergence rate, allowing consensus to be achieved with fewer sub-iterations and greater accuracy. Additionally, we observed that solving the distributed optimization problem for each node's covariance matrix is a static optimization task, eliminating the need for sub-iterations.
We demonstrated the stability of the consensus-based algorithm by modeling the update system as a discrete-time dynamical system. Furthermore, we showed all local estimators are unbiased.
%that the state estimates asymptotically converge to the true state of the dynamical system in the expected sense.
\bibliographystyle{IEEEtran}
\bibliography{refs}

% \section*{Some Comments}
% \begin{itemize}
%     \item Distributed filtering in the SLR framework
%     \item Simo was suggesting that we need to do the likelihood approximation (something like that, not sure) and then do the stability analysis
%     \item iterative SLR can also be implemented
%     \item How we reach to Eq. (25) from Eq. (24)? 
% \end{itemize}
% $b = \lambda^\top L \mathcal{X}$, $\frac{\partial b}{\partial \lambda} = L \mathcal{X}$\\
% $\frac{\partial b}{\partial \mathcal{X}} = \frac{\partial (\lambda^\top L) \mathcal{X}}{\partial \mathcal{X}} = (\lambda^\top L)^\top = L^\top \lambda$;
% Here, $L$ is a symmetric matrix. 

\end{document}